\documentclass[preprint,11pt,authoryear]{elsarticle}

\usepackage{mathtools}
\usepackage{booktabs} 
\usepackage{amsmath,tikz,pgfplots}
\usetikzlibrary{decorations.text,calc,arrows.meta}
\usetikzlibrary{calc,intersections}

\usepackage{bm}
\usepackage{varwidth,url}              
\usepackage{algorithmicx,caption}      
\usepackage{algorithm}                 
\usepackage[noend]{algpseudocode}      
\makeatletter                          
\def\BState{\State\hskip-\ALG@thistlm} 
\makeatother                           

\makeatletter
\newcommand*\bigcdot{\mathpalette\bigcdot@{1.1}}
\newcommand*\bigcdot@[2]{\mathbin{\vcenter{\hbox{\scalebox{#2}{$\m@th#1\bullet$}}}}}
\makeatother

\usepackage{bm}
\usepackage{mathrsfs}

\usepackage{amsthm}

\newcommand{\new}[1]{{{#1}}}

\usepackage{calligra,amssymb}
\usepackage{amsmath,amsfonts}

\theoremstyle{definition}
 \newtheorem{theorem}{Theorem}
 \newtheorem{proposition}[theorem]{Proposition}   
 \newtheorem{lemma}[theorem]{Lemma}   
 \newtheorem{coroll}[theorem]{Corollary}   
 \newtheorem{defn}[theorem]{Definition}   
 \newtheorem{example}[theorem]{Example}   

\newtheorem{remark}[theorem]{Remark}

\definecolor{DarkBlue}{rgb}{0,0.1,0.55}

\numberwithin{equation}{section}

\newcommand{\zarl}{\text{\calligra l}\,}

\newcommand{\hide}[1]{}
\newcommand{\zar}[1]{\text{Zar}(#1)}
\newcommand{\zz}[1]{\mathscr{L}(#1)}

\newcommand {\junk}[1]{}

\newcommand {\Z}  {\mathbb{Z}}
 \newcommand {\N}         {\mathbb{N}}

\newcommand {\ZZ} {{\rm Z}}
\newcommand {\RR} {{\mathcal R}}

\newcommand {\la}   {{\langle}}
\newcommand {\ra}   {{\rangle}}
\newcommand {\eps} {{\varepsilon}}
\newcommand {\ext} {{\mathrm{ext}}}

\newcommand {\Id} {\mbox{\rm Id}}

\newcommand{\rank}{\mathrm{rank}}

\def\addots{\mathinner{\mkern1mu
\raise1pt\vbox{\kern7pt\hbox{.}}
\mkern2mu\raise4pt\hbox{.}\mkern2mu
\raise7pt\hbox{.}\mkern1mu}}

\newcommand{\x}{\mathbf{x}}

\def\bmx{{{x}}}

\newcommand{\simone}[1]{{#1}}

\newcommand{\calV}{\mathcal{V}}

\def\ZZ{{\mathbb{Z}}}
\def\QQ{{\mathbb{Q}}}
\def\RR{{\mathbb{R}}}
\def\CC{{\mathbb{C}}}

\def\SS{{\mathbb{S}}}
\def\Id{{\mathbb{I}}}

\def\spec{\mathscr{S}}

\newcommand{\zeroset}[1]{{Z(#1)}}

\def\setD{{\mathcal{D}}}

\def\zarB{{\mathscr{B}}}

\def\rank{{\rm rank}\,}
\def\deg{{\rm deg}}

\def\mymid{{\,\,:\,\,}}

\usepackage{amssymb}

\journal{Journal of Symbolic Computation}

\begin{document}

\begin{frontmatter}



\title{{\bf Exact algorithms for semidefinite programs with degenerate feasible set}}

\author{Didier Henrion}
\address{
LAAS-CNRS, Universit\'e de Toulouse, France,
Faculty of Electrical Engineering, Czech Technical Universy in Prague, Czechia.
{henrion@laas.fr}
}

\author{Simone Naldi}
\address{
{Univ. Limoges, XLIM, UMR 7252}
{F-87000 Limoges}
{France} {simone.naldi@unilim.fr}
}

\author{Mohab Safey El Din}
\address{
  {Sorbonne Universit\'e, \textsc{CNRS},\\
    Laboratoire d'Informatique de Paris~6, \textsc{LIP6},
    \'Equipe \textsc{PolSys}}
  {4 place Jussieu}
  {F-75252, Paris Cedex 05} 
  {France}. {mohab.safey@lip6.fr}
}


  
\begin{abstract}
Given symmetric matrices $A_0, A_1, \ldots, A_n$ of size $m$ with rational entries, the set of real vectors $x=(x_1,\ldots,x_n)$ such that the matrix $A_0+x_1 A_1+\cdots+x_n A_n$ has non-negative eigenvalues is called a spectrahedron. Minimization of linear functions over spectrahedra is called semidefinite programming. Such problems appear frequently
  in control theory and real algebra, especially in the context of nonnegativity certificates for
  multivariate polynomials based on sums of squares.

  Numerical software for semidefinite programming are mostly based on
  interior point methods, assuming non-degeneracy properties such
  as the existence of an interior point in the spectrahedron. In this
  paper, we design an exact algorithm based on symbolic homotopy for
  solving semidefinite programs without assumptions on the feasible
  set, and we analyze its complexity. Because of the exactness of the
  output, it cannot compete with numerical routines in practice. However, we
  prove that solving such problems can be done in polynomial time
  if either $n$ or $m$ is fixed.
\end{abstract}

\begin{keyword}
{Semidefinite programming, Polynomial optimization, Exact computation, Homotopy}
\end{keyword}

\linespread{1}

\end{frontmatter}

\section{Introduction}\label{sec:intro}

Let $A_0, A_1, \ldots,A_n$ be symmetric matrices of size $m$ with entries in the field $\QQ$ of
rational numbers. The goal of this article is to design algorithms for
solving the semidefinite programming (SDP) problem
\begin{equation}
  \label{originalSDP}
  \inf \ell(x) \,\,\,\, \text{ s.t. } \, x \in \spec(A)
\end{equation}
where $\ell(x) = \ell_1x_1 + \cdots + \ell_n x_n$ is a linear function of a vector $x=(x_1, \ldots, x_n)$ of variables,
and $\spec(A)$ is the solution set in $\RR^n$ of the linear matrix
inequality (LMI)
\begin{equation}
  \label{originalLMI}
A(x) \coloneqq A_0+x_1A_1+\cdots+x_n A_n \succeq 0.
\end{equation}
The above constraint means that
$A(x)$ is positive semidefinite, that is, that all its eigenvalues are
non-negative. The set $\spec(A)$, called {\em spectrahedron}, is
a convex and basic semi-algebraic set, an affine section of the cone of
positive semidefinite matrices.

LMIs and SDP appear frequently in applications, {\it e.g.} for stability queries in control theory \cite{boyd1994linear}. They also appear as a central object in convex algebraic geometry
and real algebra for computing certificates of non-negativity based on sums of squares
\cite{lasserreMoments,blekherman2013semidefinite}.

Even though SDP can be solved in polynomial time to a fixed accuracy via the ellipsoid algorithm, the complexity status of this problem in the Turing or in the real numbers model is still an open question in computer science \cite{ramana1997exact,allamigeon2018solving}. On the other hand, very few algebraic methods that can represent an alternative to classical approaches from optimization theory have been developed for SDP.

In this paper, we aim at designing a symbolic algorithm for solving
the SDP problem \eqref{originalSDP}, without any assumption on the feasible
set $\spec(A)$, but with genericity assumptions on the objective
function $\ell$. Our algorithm returns an exact algebraic representation of an optimal solution.

\subsection{State of the art}
\label{ssec:stateart}

Numerical methods have been developed for solving SDP problems, the most efficient of which are based
on interior point methods (IPM) \cite{nestnem}. The core algorithmic component is Newton's method for following an algebraic primal-dual curve called {\it central
path}, whose points $(x_\mu,y_\mu)$ are real solutions to the quadratic semi-algebraic problems
\begin{equation}\label{centralpath}
A(x)Y(y) = \mu\,\Id_m \,\,\,\,\,\,\,\,\,\, A(x) \succeq 0 \,\,\,\,\,\,\,\,\,\, Y(y) \succeq 0.
\end{equation}
In the above problem, $Y(y)$ is a square matrix lying in a space of matrices dual to that of $A(x)$. For small but positive $\mu$, when the LMI has strictly feasible solutions, the points $x_\mu$ lie in the interior of $\spec(A)$. They converge to a boundary point for $\mu \rightarrow 0^+$. Logarithmic barrier functions are available for SDP, and their regularity properties are used to control the complexity of Newton's method  for solving \eqref{originalSDP} when $\spec(A)$ has interior points.

Let us remark that in several situations $\spec(A)$ has empty interior, for instance when $\spec(A)$ consists of sums-of-squares certificates of a polynomial with rational coefficients that does not admit rational certificates, see \cite{scheidRational} for a class of such examples. It is worth to recall that in the absence of interior points, variants of the IPM have been designed, such as infeasible-start IPMs (\cite{kojima1998local}), most of which are based on homotopic deformations. In this work, we use homotopy to deal with degenerate SDP in an exact setting.

The output of our algorithm is an exact representation of a feasible solution, whenever it exists, through rational univariate representation: the entries of the solution vector are rational functions of an algebraic number which is encoded by its minimum polynomial and by a separating interval. Let us also mention that an intermediate step of the algorithm consists of the computation of a similar representation (but bivariate) of a homotopy curve that contains such solution.

Even in non-degenerate situations (existence of an interior point), it is hard to estimate the degree of the central path (that represents a complexity measure for path-following methods) in practical situations and explicit examples of central paths with exponential curvature have been computed, see \cite{allamigeon2018log}. For these examples, the convergence rate of the IPM is slower than for the generic case. In contrast, we are able to give explicit upper bounds for the degree of the algebraic curve encoded by the output of our algorithm.

The several existing variants of the IPM are implemented in software running in finite precision, e.g. SeDuMi (\cite{SeDuMi}), SDPT3 (\cite{toh1999sdpt3}) and MOSEK (\cite{andersen2000mosek}). The expected running time is essentially polynomial in $n, m, \log(\eta^{-1})$ (where $\eta$ is the precision) and in the bit-length of the input \cite[Ch.1, Sec.1.4]{anjos2012introduction}. Whereas these numerical routines run quite efficiently on large instances, they may fail on degenerate situations, even on medium or small size problems. This has motivated for instance the development of floating point libraries for SDP working in extended precision, see \cite{joldes2017implementation}.

Symbolic computation has been used in the context of SDP to tackle several related problems. First, it should be observed that $\spec(A)$ is a semialgebraic set in $\RR^n$ defined by sign conditions on the coefficients of the characteristic polynomial $t \mapsto det(t\,\Id_m-A(x))$. Hence, classical real root finding algorithms for semialgebraic sets such as \cite{BPR98, BaPoRo06,Safey07,BGHS14} can be used for exact SDP. These algorithms solve SDP in time $m^{O(n)}$. Algorithms for solving diophantine problems on LMI have been developed in \cite{QSZ13} and \cite{SZ10}. The algorithm in \cite{porkolab1997complexity} can solve general semidefinite systems in simply-exponential time, which reduces to polynomial time when the matrix size or the dimension is fixed; the drawback is that the complexity is based on Renegar's quantifier elimination, and hence the constant hidden in the bound is too large for these algorithms to be used in practice.

More recently, algorithms have been designed for solving exactly {\em generic} LMI in \cite{HNS2015c, spectra} and {\em generic} rank-constrained SDP in \cite{naldiIssac2016integral}, with runtime polynomial in $n$ (the number of variables, or equivalently the dimension of the affine section defining $\spec(A)$) if $m$ (the size of the matrix) is fixed. Because of the high degrees needed to encode the output ({\it cf.} \cite{stu}), they cannot in general compete with numerical software. On small size problems they may however offer a nice complement to these techniques in situations where numerical issues are encountered. In both cases, {genericity} assumptions on the input are required. This means that for some special problems (lying in some Zariski closed subset of the space spanned by the entries of matrices $A_i$), these algorithms cannot be applied.

\subsection{Outline of the main contributions}
In this paper, we remove the genericity assumptions on the spectrahedron that were required in our previous work \cite{HNS2015c}, and we show that optimization of generic linear functions over spectrahedra can be performed without significant extra cost from the complexity viewpoint.

Our precise contributions are as follows:
\begin{itemize}
\item we design an algorithm for solving the SDP in \eqref{originalSDP} without any assumption on the defining matrix $A(x)$, but with genericity assumptions on the objective function;
\item we prove that the number of arithmetic operations performed by this algorith is polynomial in $n$ when $m$ is fixed, and viceversa;
\item we report on examples showing the behaviour of the algorithm on small-size but degenerate instances.
\end{itemize}
The main tool is the construction of a homotopy acting on the matrix representation $A(x)$ rather than on the classical complementarity conditions as in \eqref{centralpath}. This allows us to preserve the LMI structure along the perturbation. Let us remark that a genericity assumption on the objective function implies that the dual feasible set is a generic translation of a possibly degenerate spectrahedron.

We use techniques from real algebraic geometry similar to those in \cite{HNS2015c}, based on transversality theory \cite{demazure2013bifurcations}, to prove genericity properties of the perturbed systems. We also investigate closedness properties of linear maps restricted to semi-algebraic sets in a more general setting in Section \ref{sec:prelim}, generalizing similar statements for real algebraic sets in \cite{SaSc03, HNS2014}.

\simone{This work is an extended version of the paper \cite{issac18HNS}, published in the Proceedings of the International Symposium on Symbolic and Algebraic Computation (ISSAC 2018). With respect to this former version, we have clarified some aspects concerning the properness of the restriction of linear functions to spectrahedra and their homotopic perturbations (essentially contained in Lemma \ref{lemma:3}); moreover we discuss in Section \ref{sec:examples} an application of our method to the solution of primal-dual SDP problems with positive duality gap.}

\subsection{General notation}
For a matrix of polynomials $f \in \RR[x]^{s \times t}$ in $x=(x_1,\ldots\allowbreak,x_n)$, we denote by $\zeroset{f}:=\{x \in {\mathbb C}^n : f(x) = 0\}$ the complex algebraic set defined by the zero locus of $f$. If $f \in \RR[x]^{s}$, the Jacobian matrix of $f$ is denoted by $Df \coloneqq \left({\partial f_i}/{\partial x_j}\right)_{ij}$. A set $S \subset \RR^n$ defined by sign conditions on a finite list of polynomials is called a basic semi-algebraic set, and a finite union of such sets is called a semi-algebraic set.

Let $\SS^m(\QQ)$ denote the space of symmetric matrices  of size $m$ with entries in $\QQ$, and let $\SS_+^m(\QQ)$ denote the cone of positive semidefinite matrices in $\SS^m(\QQ)$. Let $A(x) := A_0+\sum_{i=1}^n x_i A_i$, with $A_i \in \SS^m(\QQ)$. One can associate to $A$ the hierarchy of algebraic sets
$$
\setD_r(A) := \{x \in \RR^n \mymid \rank A(x) \leq r\}, \,\,\,\, r=1,\ldots,m-1
$$
defined by minors of $A(x)$ of a fixed size. The set $\setD_r$ is called a determinantal variety. We recall the definition of incidence variety in the context of semidefinite programming, introduced in \cite{HNS2015c}. For $r \in \{1,\ldots,m-1\}$, let $Y = Y(y)$ be a $m \times (m-r)$ matrix of unknowns $y_{i,j}$. Let $\iota\subset\{1,\ldots,m\}$ be a subset of cardinality $m-r$, and $Y_\iota$ the submatrix of $Y$ corresponding to rows indexed by $\iota$. The {\it incidence variety} for $\setD_r(A)$ is the algebraic set
$$
\calV_{r,\iota}(A) := \{ (x,y) \in \CC^n \times \CC^{m(m-r)} \mymid A(x)Y(y)=0, Y_\iota = \Id_{m-r} \}.
$$

Let $B \in \SS^m(\QQ)$ and $\varepsilon \in [0,1]$. In this paper, we consider a 1-parameter family of linear matrices
  $$
  A(x)+\varepsilon B = (A_0+\varepsilon B)+\sum_i x_i A_i
  $$
  perturbing $A(x)$ in direction $B$.

  \section{Preliminaries}\label{sec:prelim}

In this section, we prove some results of topological nature on
spectrahedra and their deformations. Before doing that, we need to
recall basics about infinitesimals and Puiseux series rings. More
details can be found in \cite{BaPoRo06}.

An infinitesimal $\eps$ is a positive element which is transcendental
over $\RR$ and smaller than any positive real number. The field of
Puiseux series
$$
\RR\langle \eps\rangle:=\left\{\sum_{i\geq i_0}a_i \eps^{i/q} \: :\:  i_0\in \Z, q \in \N\setminus\{0\}, a_i \in \RR\right\}
$$
is real closed, \cite[Ex.1.2.3]{BoCoRo98}. An element $z = \sum_{i\geq i_0}a_i \eps^{i/q}$ is bounded over $\RR$ if $i_0 \geq 0$. In that case, one says that its limit when $\eps$ tends to $0$ is $a_0$ and we write it $\lim_\eps z$. The operator $\lim_\eps \colon \RR\langle \eps\rangle \to \RR$ is a ring homomorphism, and we extend it over $\RR\langle\eps\rangle^n$ coordinatewise. Also given a subset $Q \subset \RR\langle\eps\rangle^n$, we denote by $\lim_\eps Q$ the subset of $\RR^n$ of points which are the images by $\lim_\eps$ of bounded elements in $Q$.

Given a semialgebraic set $S\subset \RR^n$ defined by a
semialgebraic formula with coefficients in $\RR$, 
$\ext(S, \RR\langle\eps\rangle)$ denotes the solution set of that formula in
$\RR\langle\eps\rangle^n$. Incidentally, we often abuse of notation, letting $S$ denote
both solutions sets.

For a linear pencil $A(x) = A_0+x_1 A_1+\cdots+x_nA_n$ of $m\times m $
symmetric linear matrices and a $m \times m$ positive definite matrix
$B$, we consider the spectrahedron $\spec(A+\eps B)$ in
$\RR\langle \eps\rangle^n$.

Our first result relates
$\spec(A)\subset \RR^n$ with
$\spec(A+\eps B)\subset \RR\langle\eps\rangle^n$.

\begin{lemma}\label{lemma:1}
  Using the above notation, $\spec(A)$ is included in (the interior of)
  $\spec(A+\eps B)$.
\end{lemma}

\begin{proof}
  \new{If $\spec(A)= \emptyset$, there is nothing to prove.}
  Let $\bmx^* \in \spec(A)$. By definition of positive semide\-fi\-ni\-te\-ness, for
  any vector ${v}\in \RR^m$, ${v}^t A(\bmx^*){v} \geq 0$. Since
  $\eps$ is a positive infinitesimal and $B$ is positive definite, we
  deduce that for any vector ${v}\in \RR^m\setminus \{0\}$,
  $0<{v}^t A(\bmx^*){v}+ {v}^t \eps B{v} = {v}^t(A(\bmx^*)+\eps B){v}$.
  We deduce that $A+\eps B$ is positive definite at $\bmx^*$, hence
  $\bmx^*$ is in (the interior of) $\spec(A+\eps B)$, as requested.
\end{proof}

Further, we identify the set of linear forms $\ell := \ell_1 x_1 +\cdots +\ell_n x_n$ with $\CC^n$, the linear form $\ell$ being identified to the point $\ell_1, \ldots, \ell_n$. By a slight abuse of notation we also denote by $\ell$ the map $\bmx \mapsto \ell(\bmx)$.

\begin{lemma}\label{lemma:sas}
  Let $\mathbf{R}$ be a real closed field, $\mathbf{C}$ be an
  algebraic closure of $\mathbf{R}$ and $S\subset \mathbf{R}^n$ be a
  closed semialgebraic set. There exists a non-empty Zariski open set
  $\zz{S}\subset \mathbf{C}^n$ such that for
  $\ell\in \zz{S}\cap\mathbf{R}^n$, $\ell(S)$ is closed for the
  Euclidean topology.
\end{lemma}

\begin{proof}
  Our proof is by induction on the dimension of $S$. When $S$ has
  dimension $0$, the statement is immediate. 

  We let now $d \in \N\setminus\{0\}$, assume that the statement holds for
  semialgebraic sets of dimension less than $d$ and that $S$ has
  dimension $d$. By \cite[Th.2.3.6]{BoCoRo98}, it can be partitioned as a
  finite union of closed semialgebraically connected semialgebraic
  manifolds $S_1, \ldots, S_N$. Note that each $S_i$ is semialgebraic.
  We establish below that there exist
  non-empty Zariski open sets $\zz{S_i}\subset \mathbf{C}^n$ such
  that for $\ell\in \zz{S_i}\cap \mathbf{R}^n$, $\ell(S_i)$ is closed
  for the Euclidean topology. Taking the intersections of those
  finitely many non-empty Zariski open set is then enough to define
  $\zz{S}$.

  Let $1\leq i \leq N$. If the dimension of $S_i$ is less than $d$, we
  apply the induction assumption and we are done. Assume now that
  $S_i$ has dimension $d$. Let $V\subset \mathbf{C}^n$ be the Zariski
  closure of $S_i$ and $C$ be the semialgebraically connected
  component of $V\cap \mathbf{R}^n$ which contains $S_i$.  By
  \cite[Prop.17]{HNS2014}, there exists a non-empty Zariski open
  set $\Lambda_{1, i}\subset \mathbf{C}^n$ such that for
  $\ell\in \Lambda_{1, i}\cap \mathbf{R}^n$, $\ell(C)$ is closed.

  By definition of $C$ and using \cite[Ch.2.8]{BoCoRo98}, $C$ has
  dimension $d$, as $S_i$. We denote by $T_i\subset \mathbf{R}^n$ the
  boundary of $S_i$. Observe that it is a closed semialgebraic set of
  dimension less than $d$ \cite[Ch.2.8]{BoCoRo98}. Using the induction
  assumption, we deduce that there exists a non-empty Zariski open set
  $\Lambda_{2, i}\subset \mathbf{C}^n$ such that for
  $\ell\in \Lambda_{2, i}\cap \mathbf{R}^n$, $\ell(T_i)$ is closed.
  We claim that one can define $\zz{S_i}$ as the intersection
  $\Lambda_{1, i}\cap \Lambda_{2, i}$, i.e. for
  $\ell \in \zz{S_i}\cap \mathbf{R}^n$, $\ell(S_i)$ is closed.

  Indeed, assume that the boundary of $\ell(S_i)$ is not empty
  (otherwise there is nothing to prove) and take $a$ in this
  boundary. Without loss of generality, assume also that for all
  $\bmx\in S_i$, $\ell(\bmx)\geq a$. We need to prove that
  $a \in S_i$.

  Assume first that for all $\eta>0$, $\ell^{-1}([a, a+\eta])$ has a
  non-empty intersection with $T_i$. Since $\ell(T_i)$ is closed by
  construction, we deduce that there exists $\bmx\in T_i$ such that
  $\ell(\bmx)=a$. Since $S_i$ is closed by construction and $T_i$ is
  its boundary, we deduce that $\bmx\in S_i$ and then that
  $a \in \ell(S_i)$.

  Assume now that for some $\eta>0$, $\ell^{-1}([a, a+\eta])$ has an
  empty intersection with $T_i$. Then, we deduce that
  $\ell^{-1}([a, a+\eta])\cap S_i = \ell^{-1}([a, a+\eta])\cap C$.
  Besides, since $\ell(C)$ is closed, there exists $\bmx\in C$ such
  that $\ell(\bmx)=a$. Because,
  $\ell^{-1}([a, a+\eta])\cap S_i = \ell^{-1}([a, a+\eta])\cap C$, we
  deduce that $\bmx\in S_i$ which ends the proof.
\end{proof}

\begin{coroll}\label{lemma:2}
  Let $A(x)$ be as above.
  Then there is a non-empty Zariski open set $\zarl_1\subset \CC^n$
  such that for $\ell\in \zarl_1\cap \RR^n$, $\ell(\spec(A))$ is closed for the Euclidean
  topology of $\RR^n$.
\end{coroll}
\begin{proof}
  Since $\spec(A) \subset \RR^n$ is a closed semialgebraic
  set in $\RR^n$, one can directly apply Lemma~\ref{lemma:sas} to define the
  non-empty Zariski open set $\zarl_1 \subset \CC^n$ satisfying the claimed
  property.
\end{proof}

\simone{
\begin{defn}
  Let $M,N \subset \RR\langle \eps\rangle^n$. We say that $M$ is $\eps-$near
  to $N$ if $\forall \, y \in M$ there exists $x \in N$ such that $\lim_\eps (x-y)=0$.
\end{defn}
In the following, we often abuse of notation concerning the extension of semialgebraic sets from
$\RR^n$ to $\RR\la\eps\ra^n$ and of semialgebraic functions defined over the two fields.
}

\begin{lemma}\label{lemma:3}
 \simone{Let $\ell$ be such that $\ell(\spec(A))$ is closed.}
  \begin{enumerate}
  \item[1.]
    \simone{
    For $B \succ 0$, $\spec(A+\eps B)$ is $\eps-$near to $\spec(A)$ and $\ell(\spec(A+\eps B))$ is $\eps-$near to $\ell(\spec(A))$.
    }
  \item[2.]
    Let $B \succ 0$. Assume that there exists $\bmx^*\in \spec(A)$ such that $\ell(\bmx^*)$
    lies on the boundary of $\ell(\spec(A))$. Then, there exists $\bmx^*_\eps
    \in \spec(A+\eps B)$ such that $\ell(\bmx^*_\eps)$ lies on the boundary
    of $\ell(\spec(A+\eps B))$ and $\lim_\eps \bmx^*_\eps = \bmx^*$.
    Conversely, if $\bmx^*_\eps \in \spec(A+\eps B)$ lies on the
    boundary of $\ell(\spec(A+\eps B))$, $\spec(A) \neq \emptyset$, and $\lim_\eps \bmx^*_\eps$
    exists, then $\ell(\lim_\eps \bmx^*_\eps)$ lies on the boundary of $\ell(\spec(A))$.
  \item[3.]
    \simone{
      If the restriction of $\ell$ to $\spec(A)$ is proper, and if $B \succ 0$, then the restriction of $\ext(\ell,\RR\la\eps\ra^n)$ to $\spec(A+\eps B)$ is proper.
    }
  \end{enumerate}
\end{lemma}
\begin{proof}
    \simone{ {\it (1)}.
      For $r\in \RR$ positive, let $B(x, r)$ be the ball of radius $r$
      centered at $x$.
      First, we prove that $\spec(A+\eps B)$ is $\eps-$near to $\spec(A)$.
      This is true if and only if
      $$
      \forall r>0 \, \forall x_\eps \in \spec(A+\eps B), B(x_\eps,r) \cap \spec(A) \neq \emptyset.
      $$
      By the Tarski-Seidenberg Principle (see \cite[Lem. 3.2.1]{BaPoRo06}) this is equivalent
      to the formula
      $$
      \forall r>0 \, \exists e' \in \RR_+ \, \forall 0 < e < e' \, \forall x_e \in \spec(A+e B), B(x_e,r) \cap \spec(A) \neq \emptyset
    $$
    defined over $\RR$. The latter formula is true due to the continuity of the eigenvalue functions.
    Next, since $\ell$ is a continuous function, defined over $\RR$, one deduces that $\ell(\spec(A+\eps B))$ is $\eps-$near to $\ell(\spec(A))$.
  }

  {\it (2)}.
  Recall that $\spec(A)$ is contained in $\spec(A+\eps B)$
  (Lemma~\ref{lemma:1}) and by point 1 that that $\spec(A+\eps B)$ is
  $\eps-$near to $\spec(A)$. 
  This implies that there exists $\rho_\eps$ on the boundary of
  $\ell(\spec(A+\eps B) \cap B(\bmx^*, r))$ and which is
  $\eps-$near to $\ell(\bmx^*)$.
  Since $\spec(A+\eps B) \cap B(\bmx^*,r)$ is closed and
  bounded, $\ell(\spec(A+\eps B) \allowbreak\cap \allowbreak B(\bmx^*, r))$ is closed
  for the Euclidean topology. Then, there exists
  $\bmx^*_\eps\in \spec(A+\eps B)\cap B(\bmx^*,r)$ such
  that $\ell(\bmx^*_\eps) = \rho_\eps$. Since this is true for any
  $r\in \RR$ positive, we deduce the equality $\lim_\eps \bmx^*_\eps = \bmx^*$.

  {Conversely, suppose that $\bmx^*_\eps \in \spec(A+\eps B)$  is such
    that $\ell(\bmx^*_\eps)$ lies
    on the boundary of $\ell(\spec(A+\eps B))$. Hence $\ell(\bmx^*_\eps)$
    minimizes $\ell$ on $\spec(A+\eps B)$. Let $y \in \spec(A)$.
    From Lemma \ref{lemma:1}, we know that $y \in \spec(A+\eps B)$.
    Since orders are preserved under limit, by the continuity of
    $\ell$, and applying \cite[Lem. 3.2.1]{BaPoRo06}, we get that
    $$
    \ell(x^*) = \ell(\lim_\eps\bmx^*_\eps) = \lim_\eps \ell(\bmx^*_\eps) \leq \lim_\eps \ell(y) = \ell(y).
    $$
    Since $y$ was arbitrary, we deduce that $x^*$ minimizes $\ell$ on $\spec(A)$,
    hence $\ell(x^*)$ lies on the boundary of $\ell(\spec(A))$.}

  \simone{ {\it (3)}.
    Denote by $\ell_1$ the restriction of $\ell$ to $\spec(A)$, and by $\ell_2$ the restriction of $\ext(\ell,\RR\la\eps\ra^n)$ to $\spec(A+\eps B)$. Let $U \subset \overline{\ell_2(\spec(A+\eps B))}$ be closed and bounded in $\RR\la\eps\ra^n$, and let $V = U \cap {\ell_1(\spec(A))}$. Since $\ell_1$ is proper, then $\ell_1(\spec(A))$ is closed, hence $V$ is closed and bounded in $\RR^n$, hence compact. By the properness of $\ell_1$, we deduce that $\ell_1^{-1}(V) = \ell_1^{-1}(U \cap {\ell_1(\spec(A))}) = \ell_1^{-1}(U) \cap \ell_1^{-1}({\ell_1(\spec(A))}) = \ell_1^{-1}(U) \cap \spec(A) = \ell_1^{-1}(U)$ is compact, hence closed and bounded. Remark that the same argument used in the proof of point 1 implies that $\ell_2^{-1}(U)$ is $\eps-$near to $\ell_1^{-1}(U)$, hence that it is bounded. Moreover, by the continuity of $\ell_2$, we deduce that $\ell_2^{-1}(U)$ is closed. We conclude that $\ell_2$ is proper.
    }
\end{proof}

\simone{
  \begin{remark}
    Note that in case $\spec(A)$ is compact, the properness assumption made in Lemma \ref{lemma:3} is trivially satisfied.
  \end{remark}
}


\section{Homotopy for semidefinite systems}
\label{sec:homotopy}
We consider the original linear matrix inequality $A(x) \succeq 0$ and its solution set $\spec(A)$. In this section, we prove that regularity properties can be assumed after the deformation of $\spec(A)$ described in the previous sections.

\subsection{Regularity of perturbed incidence varieties}\label{sec:regularity}
Let $B \in \SS^m(\QQ)$ and $\varepsilon \in [0,1]$.
We say that $A+\varepsilon B$ is {\it regular} if, for every $r=1,\ldots,m$ and $\iota \subset \{1,\ldots,m\}$ with $\sharp\iota=m-r$, the algebraic set $\calV_{r,\iota}(A+\varepsilon B)$ is smooth and equidimensional, of co-dimension $m(m-r)+\binom{m-r+1}{2}$ in $\CC^{n+m(m-r)}$.

The following proposition states that such a property holds almost everywhere if the perturbation follows a generic direction.

\begin{proposition}\label{prop:smoothness}
  There exists a non-empty Zariski open set $\zarB_1 \subset \SS^m(\CC)$ such that, for all $r \in \{0,\ldots,m\}$, $\iota\subset\{1,\ldots,m\}$ with $\sharp\iota=m-r$, and for $B \in \zarB_1 \cap \SS^m(\QQ)$, the following holds. For every $\varepsilon \in (0,1]$, out of a finite set, the matrix $A+\varepsilon B$ is regular.
  \end{proposition}
  \begin{proof}
    We suppose w.l.o.g. that $r$ is fixed and $\iota = \{1,\ldots,m-r\}$.
    Let $\mathfrak{B}$ be a $m \times m$ symmetric matrix of variables. For $\varepsilon \in (0,1]$, we get a matrix $A(x)+\varepsilon \mathfrak{B}$, which is bilinear in the two groups of variables $x,\mathfrak{B}$.
    Let $f^{(\varepsilon)} = f^{(\varepsilon)}(x,y,\mathfrak{B})$ be the polynomial system given by the $(i,j)-$en\-tri\-es of $(A+\varepsilon \mathfrak{B})Y$ with $i \geq j$, and by all entries of $Y_\iota-\Id_{m-r}$, that is $\zeroset{f^{(\varepsilon)}}=\zeroset{(A+\varepsilon \mathfrak{B})Y,Y_\iota-\Id_{m-r}}$, after the reduction provided in \cite[Lemma 3.2]{HNS2015c}. One can thus assume that $\sharp f^{(\varepsilon)} = m(m-r)+\binom{m-r+1}{2}$.

    We now proceed with a transversality argument. Consider the map (with abuse of notation)
    \[
      \begin{array}{lrcc}
        f^{(1)} : &  \CC^{n} \times \CC^{m(m-r)} \times \CC^{\binom{m+1}{2}} & \longrightarrow & \CC^{m(m-r)+\binom{m-r+1}{2}} \\
                &  (x,y,\mathfrak{B}) & \longmapsto & f^{(1)}(x,y,\mathfrak{B}).
      \end{array}
    \]
    We claim that $0$ is a regular value of the map $f^{(1)}$ (the claim is proved in the last paragraph). This implies by Thom's Weak Transversality \cite[Prop. B.3]{din2013nearly} that there is a Zariski open set $\zarB_{r,\iota} \subset  \SS^m(\CC)$ such that, if $B \in \zarB_{r,\iota}$, then $0$ is a regular value of the section map $(x,y) \mapsto f^{(1)}(x,y,{B})$.

    We define $\zarB_1 \coloneqq \cap_r \cap_\iota \zarB_{r,\iota}$, which is a finite intersection of non-empty Zariski open sets, hence Zariski open and non-empty.
    Now, for a fixed $B \in \zarB_1$, consider the line $t B$, $t \in \RR$, in $\SS^m(\CC)$. Let $F_1 \in \CC[\mathfrak{B}]$ be the generator of the ideal of all polynomials vanishing over the algebraic hypersurface $\SS^m(\CC) \setminus \zarB_1$. Then, since $B \in \zarB_1$ by construction, $t \mapsto F_1(t B)$ does not vanish identically, hence it vanishes exactly $\deg\,F_1$ many times (counting multiplicities). We deduce that, $\varepsilon B \in \zarB_1$ except for finitely many values of $\varepsilon$. We conclude that for all $r$ and $\iota$, $\calV_{r,\iota}(A+\varepsilon B)$ is smooth and equidimensional of co-dimension $\sharp f^{(\varepsilon)} = m(m-r)+\binom{m-r+1}{2}$, for $\varepsilon \in (0,1]$ except for finitely many values.

    We prove now our claim. It follows by an argument similar to the proof of \cite[Prop.3.4]{HNS2015c}. Consider the derivatives of polynomials in $f^{(1)}(x,y,\mathfrak{B})$ with respect to the $(i,j)-$entries of $\mathfrak{B}$, with either $i \leq m-r$ or $j \leq m-r$, and those with respect to $y_{i,j}$ with $i \in \iota$. It is straightforward to check that these form a maximal submatrix of the jacobian matrix $D f^{(1)}$ whose determinant is non-zero, proving that $0$ is a regular value of $f^{(1)}$.
  \end{proof}

  \subsection{Critical points on perturbed linear matrix inequalities}
  \label{ssec:pert:spec}
  
  Let $B \in \SS^m(\QQ)$ and let $A+\varepsilon B$ be the perturbed linear pencil defined above. For a fixed $\varepsilon<1$, we consider the stratification of the hypersurface $\zeroset{\det(A+\varepsilon B)}$ given by the varieties $\setD_r(A+\varepsilon B)$ of multiple rank defects of $A+\varepsilon B$, and their lifted incident sets $\calV_{r,\iota}(A+\varepsilon B)$.

For $r<m$ and $\iota \subset \{1,\ldots,m\}$ with $\sharp\iota=m-r$, let $c \coloneqq m(m-r)+\binom{m-r+1}{2}$. We recall from the proof of Proposition \ref{prop:smoothness} that $f^{(\varepsilon)} \in \RR[x,y]^c$ consists of the $(i,j)-$en\-tri\-es of $A^{(\varepsilon)}Y$ with $i \geq j$, and by all entries of $Y_\iota-\Id_{m-r}$. We define the {\it Lagrange system} $\text{Lag}_{r,\iota}(A+\varepsilon B)$ as follows:
\begin{equation}
 \label{lagrange}
\begin{aligned}
  f^{(\varepsilon)}_i(x,y) & = 0, \,\,\,\,\, i=1,\ldots,c \\
  \sum_{i=1}^c z_i \nabla f^{(\varepsilon)}_i(x,y) & = \left(\begin{matrix} \ell \\ 0 \end{matrix}\right)
\end{aligned}
\end{equation}
where $\ell : \RR^n \rightarrow \RR$ is linear. As in Section \ref{sec:prelim}, we abuse the notation of $\ell$, and identifying it with the vector $(\ell_1,\ldots,\ell_n) \in \RR^n$ giving $\ell(x)=\ell_1x_1+\cdots+\ell_nx_n$, hence $\ell=\nabla\ell$.

The set $\zeroset{f^{(\varepsilon)}}=\calV_{r,\iota}(A+\varepsilon B)$ is smooth for generic $B$ thanks to Proposition \ref{prop:smoothness}. Hence a solution $(x^*,y^*,z^*)$ of system \eqref{lagrange} is a critical point $(x^*,y^*)$ of the restriction of $\ell$ to $\calV_{r,\iota}(A+\varepsilon B)$, equipped with a Lagrange multiplier $z^* \in \CC^{c}$. Such a solution is called {\it of rank $r$} if $\rank\,(A(x^*)+\varepsilon B)=r$.

\begin{proposition}
  \label{prop:crit}
  There are two non-empty Zariski-open sets $\zarB_2 \subset \SS^m(\CC)$ and $\zarl_2 \, \subset \CC^n$ such that, for $B \in \zarB_2 \cap \SS^m(\QQ)$, $\ell \in \zarl_2 \, \cap \QQ^n$, and $\varepsilon\in (0,1]$ out of a finite set, the following holds.
    Suppose that $\ell$ has a minimizer or maximizer $x^*_\varepsilon$ on $\spec(A+\varepsilon B)$. The projection on the $x-$space of the union, for $\iota \subset \{1,\ldots,m\}$, $\sharp\iota=m-r$, of the solution sets of rank $r$ of system \eqref{lagrange}. It is finite and it contains $x^*_\varepsilon$.
\end{proposition}
\begin{proof}
  Let $r\leq m-1$ and $\iota \subset \{1,\ldots,m\}$.
  Recall by \cite[Th. 4]{Naldi2016357} that a minimizer or a maximizer $x^*$ for the SDP $\inf\{\ell(x) \mymid A(x)+\varepsilon B \succeq 0\}$, with $\rank\allowbreak(A(x^*)+\varepsilon B)=r$, is a critical point of the restriction of $\ell$ to $\setD_r(A+\varepsilon B)$. Moreover, \cite[Lem. 2]{Naldi2016357} implies that such critical points can be computed as the projections on the $x-$space, of the critical points of the restriction of $\ell$ to $\calV_{r,\iota}(A+\varepsilon B)$, for some $\iota$ (here we mean the extension $(x,y) \mapsto \ell(x)$ of $\ell$ to the $(x,y)-$space). Thus we only need to prove the finiteness of solutions of rank $r$ of system \eqref{lagrange}, for a generic perturbation matrix $B$ and a generic linear function $\ell$, {\it uniformly} on $\varepsilon$.

  We denote by $g^{(\varepsilon)}=z^TDf^{(\varepsilon)}-(\ell,0)^T$ (the polynomials in the second row of \eqref{lagrange}). The system $(f^{(\epsilon)},g^{(\epsilon)})$ is square, for a fixed $\epsilon$. Consider the polynomial map $(f^{(1)},g^{(1)})$ sending $(x,y,\mathfrak{B},z,\mathfrak{l})$ to $(f^{(1)}(x,y,\mathfrak{B}),g^{(1)}(x,y,\mathfrak{B},z,\mathfrak{l}))$, where $\mathfrak{B}$ and $\mathfrak{l}$ are variables for $B$ and $\ell$, of the right size. As in the proof of Proposition \ref{prop:smoothness}, for generic $B$ the rank of $Df^{(1)}$ is maximal. Hence, following {\it mutatis mutandis} the proof of \cite[Prop.3]{Naldi2016357}, we conclude that the jacobian matrix of $(f^{(1)},g^{(1)})$ has full rank at every point in $\zeroset{f^{(1)},g^{(1)}}$ of rank $r$. Hence there exist non-empty Zariski open sets $\zarB_{r,\iota} \subset \SS^m(\CC), \zarl_{r,\iota} \, \subset \CC^n$ such that if $(B,\ell) \in \zarB_{r,\iota} \times \zarl_{r,\iota}$ then system \eqref{lagrange} has finitely many solutions of rank $r$, for $\varepsilon=1$. We define $\zarB_2 \coloneqq \cap_r \cap_\iota \zarB_{r,\iota}$ and $\zarl_2 \coloneqq \cap_r \cap_\iota \zarl_{r,\iota}$ and we conclude the same disregarding $r$ and $\iota$.

  Let $F_2 \in \CC[\mathfrak{B},\mathfrak{l}]$ be the generator of the ideal of all polynomials vanishing over $(\SS^m(\CC) \times \CC^n) \setminus (\zarB_2 \times \zarl_2 \,\, )$. Then $F_2(B,\ell) \neq 0$, which implies that $t \mapsto F_2(tB,\ell)$ has finitely many roots, hence $(\varepsilon B,\ell) \in (\zarB_2 \times \zarl_2 \,\, )$ almost everywhere in $(0,1]$. We conclude the proof by defining the claimed finite set as the union of (1) the set of roots of $F_2$ and (2) the finite set constructed in Proposition \ref{prop:smoothness}.
\end{proof}

Note that the transversality techniques used in the proofs of Propositions \ref{prop:smoothness} and \ref{prop:crit} are non-constructive. Indeed they prove the existence of the \emph{discriminants} $F_1 \in \CC[\mathfrak{B}]$ and $F_2 \in \CC[\mathfrak{B},\mathfrak{l}]$, but do not construct them effectively. If we knew $F_1,F_2$ one could use separation bounds for real roots of univariate polynomials ({\it e.g.} \cite{HerHonTsi-mDMM-17}) to get upper bounds for the minimum of the finite sets: this would give an exact {\it semialgebraic} description of the homotopy curve. The output of our algorithm will be an exact representation of the {\it complex closure} of such curve.

  

  
  \subsection{The degree of the homotopy curve}
  We consider the Lagrange system \eqref{lagrange}, $r<m$ and $\iota \subset \{1,\ldots,m\}$ with $\sharp\iota=m-r$. For a given homotopy parameter $\varepsilon \in (0,1)$ out of the union of the finite sets defined in Propositions \ref{prop:smoothness} and \ref{prop:crit}, the system has finitely many solutions of rank $r$. When $\varepsilon$ converges to $0$, these solutions follow a (possibly reducible) semialgebraic curve. This can also be seen as a semialgebraic subset of dimension $1$ in $\RR\langle\eps\rangle^n$.
  We denote this curve by $\mathcal{C}_{r,\iota}$.

  Contrarily to the classical homotopy based on the central path, whose points lie in the interior of the feasible set, we have constructed homotopy curves containing optimal solutions of given rank of perturbed semidefinite programs. This allows to derive degree bounds that depend on this rank.
  
  \begin{proposition}\label{prop:deg:curve}
    Let $r,\iota$ be fixed, let $\mathcal{C}_{r,\iota}$ be the curve of solutions or rank $r$ of the Lagrange system \eqref{lagrange}, for positive small enough $\varepsilon$, and $\zar{\mathcal{C}_{r,\iota}}$ be its complex Zariski closure. Then
  $$
  \deg\,\zar{\mathcal{C}_{r,\iota}} \leq \left(1+2r(m-r)\right) \cdot \theta_1
  $$
  where
  \begin{equation}
    \label{degbound}
    \theta_1 = \sum_k \binom{c}{n-k}\binom{n}{c+k-r(m-r)}\binom{r(m-r)}{k}
  \end{equation}
  and $c \coloneqq m(m-r)+\binom{m-r+1}{2}$.
  \end{proposition}
\begin{proof}
  We first compute a polynomial system equivalent to \eqref{lagrange}. We make the substitution $Y_\iota=\Id_{m-r}$ that eliminates variables $\{y_{i,j} \mymid i \in \iota\}$ in the vector $f^{(\varepsilon)}$ defining the incidence variety $\calV_{r,\iota}(A+\varepsilon B)$, hence we suppose $f^{(\varepsilon)} \in \QQ[\varepsilon,x,\overline{y}]^c$, with $c = m(m-r)-\binom{m-r}{2}=\frac{(m-r)(m+r+1)}{2}$ and $\overline{y}=\{y_{i,j} \mymid i \not\in \iota\}$. (Indeed, $\binom{m-r}{2}$ is the number of redundancies eliminated by \cite[Lemma 3.2]{HNS2015c} recalled in the proof of Proposition \ref{prop:smoothness}.) Above we have intentionally abused of the notation of $f^{(\varepsilon)}$ and $c$. Next, the new polynomials $f_i$ do not depend on $y \setminus \overline{y}$. Hence, defining $g \coloneqq \sum_{i=1}^c z_i \nabla f^{(\varepsilon)}_i(x,\overline{y}) - (\nabla\ell,0)^T \in \QQ[\varepsilon,x,\overline{y},z]$, with $z=(z_1,\ldots,z_c)$, one has $\sharp g = \sharp x+\sharp \overline{y}=n+r(m-r)$.
  
  We conclude that the Lagrange system \eqref{lagrange} is given after reduction by the entries of $f^{(\varepsilon)}$ and $g$, that are multilinear in the three groups of variables $\xi \coloneqq (\varepsilon,x),\overline{y}$ and $z$. The multidegree with respect to $(\xi,\overline{y},z)$ is respectively
\begin{itemize}
  \item $\text{mdeg}_{(\xi,\overline{y},z)}(f^{(\varepsilon)}_i) = (1,1,0)$, for $i=1,\ldots,c$
  \item $\text{mdeg}_{(\xi,\overline{y},z)}(g_i) = (0,1,1)$, for $i=1,\ldots,n$
  \item $\text{mdeg}_{(\xi,\overline{y},z)}(g_{n+j}) = (1,0,1)$, for $j=1,\ldots,r(m-r)$
\end{itemize}
We compute below a multilinear B\'ezout bound of $\deg\,\zar{\mathcal{C}_{r,\iota}}$ (see \cite[App.H.1]{din2013nearly}). This is given by the sum of the coefficients of the polynomial
$$
P=(s_1+s_2)^c (s_2+s_3)^n (s_1+s_3)^{r(m-r)}
$$
modulo the monomial ideal $I = \langle s_1^{n+2}, s_2^{r(m-r)+1}, s_3^{c+1} \rangle$. Since the maximal admissible power modulo $I$ of $s_1$ (resp. of $s_2,s_3$) is $n+1$ (resp. $r(m-r),c$) and since $P$ is homogeneous of degree $c+n+r(m-r)$ we get
$$
P \equiv \theta_1 s_1^{n}s_2^{r(m-r)}s_3^{c}+\theta_2 s_1^{n+1}s_2^{r(m-r)-1}s_3^{c}+\theta_3 s_1^{n+1}s_2^{r(m-r)}s_3^{c-1}
$$
modulo $I$,
where $\theta_i = \theta_i(m,n,r)$ are the corresponding coefficiens in the expansion of $P$, hence the bound is $\theta_1+\theta_2+\theta_3$. Just by expanding $P$ and by solving a linear system over $\ZZ$ one gets the expression in \eqref{degbound}, within the range $0 \leq k \leq \min\{n-c+r(m-r),r(m-r)\}$. A similar formula holds for $\theta_2$ where $n-k+1$ substitutes $n-k$ in the first binomial coefficient. We deduce that
$$
\theta_2 \leq \max_k\left\{\frac{c-n+k}{n-k+1}\right\}\theta_1 \leq r(m-r) \theta_1.
$$
Moreover the expression of $\theta_3$ equals that of $\theta_2$ except for the second binomial coefficient which is smaller, hence $\theta_3 \leq \theta_2 \leq r(m-r) \theta_1$, and we conclude.
\end{proof}


Recall that the algorithm in \cite{HNS2015c}, which does not rely on homotopy, solves LMI  in the regular case, that is, under some genericity properties. In the current paper, which relies on homotopy, we do not assume these genericity properties. We expect that in degenerate situations the degree of the homotopy curve will exceed that of the univariate representation computed in the regular case. We prove that this degree gap is controlled, namely, that the extra factor is linear in $n$ and in the rank-corank coefficient $r(m-r)$.

\begin{proposition}\label{prop:comparison}
  Let $\theta = \theta(m,n,r)$ be the bound computed in \cite[Prop.5.1]{HNS2015c}. For all $r$ and $\iota$ as above
  $$
  \deg\,\zar{\mathcal{C}_{r,\iota}} \leq \left(1+\allowbreak 2r(m-r)\right) n \theta.
  $$
\end{proposition}
\begin{proof}
  Let $\theta_1$ be the expression in \eqref{degbound}. We prove that $\theta_1 \leq n \theta$ and we conclude. Indeed, let $\theta = \sum_k a_k$ and $\theta_1 = \sum_k b_k$. Then
  $$
  \frac{b_k}{a_k} = \frac{n}{c+k-r(m-r)}
  $$
  that does not exceed $n$ for all $k$. Hence $\theta_1 \leq \sum_k na_k = n \theta$.
\end{proof}

\section{Algorithm}\label{sec:algo}

\subsection{Description}\label{ssec:algo:descr}

This section contains the formal description of a homotopy-based algorithm for solving the semidefinite program in \eqref{originalSDP}, called \textsc{DegenerateSDP}.

We first define the data structures we use to represent algebraic sets of dimension $0$ and $1$ during the algorithm. A {\it zero-dimensional parametrization} of a finite set $W \subset \CC^n$ is a vector $Q=(q_0,q_1,\ldots,\allowbreak{}q_{n},q) \in \QQ[t]^{n+2}$ such that $q_0,q$ are coprime and
$$
W = \left\{ a \in \CC^n : \exists t \in \RR, a_i = \frac{q_i(t)}{q_0(t)}, q(t) = 0\right\}.
$$
Similarly a {\it one-dimensional parametrization} of a curve $\mathcal{C} \subset \CC^n$ is a vector $Q=(q_0,q_1,\ldots,q_{n},q) \in \QQ[t,u]^{n+2}$ with $q_0,q$ coprime and
$$
\mathcal{C} = \left\{ a \in \CC^n : \exists t,u \in \RR, a_i = \frac{q_i(t,u)}{q_0(t,u)}, q(t,u) = 0\right\}.
$$
Abusing notation we denote by $\zeroset{Q}$ the sets in the right part of the
previous equalities. If $Q$ is now a finite {\it list} of parametrizations, $\zeroset{Q}$ denotes
the union of $\zeroset{Q_j}$ for $Q_j$ in $Q$, \new{and every
  $x^*\in\zeroset{Q}$ is encoded by $(Q,[a_*,b_*])$, where $a_*,b_* \in \QQ$ and
  $[a_*,b_*]$ is a separating interval for the root (of the polynomial $q$ in one of the $Q_j$) that corresponds to $x^*$.}
These representations for finite sets and curves are standard in real algebraic
geometry, and are called parametrizations in the sequel. By convention, $(\,)$
is a parametrization for $\emptyset$.

We also define the following subroutines manipulating this kind of representations:
\begin{itemize}
\item
    {\sf ODP}. With input a polynomial system $f = (f_1,\ldots,f_s)$ defining a one-di\-men\-sio\-nal algebraic set $\zeroset{f}$, and a set of variables $x$, it returns a one-di\-men\-sio\-nal parametrization of the projection of $\zeroset{f}$ on the $x-$space.
\item
  {\sf CUT}. Given a one-dimensional parametrization $Q$ of the zero set $\zeroset{f}\subset\CC^{n+1}$ of polynomials $f_1,\ldots,f_s \in \QQ[\varepsilon,x]$, it returns a zero-dimensional parametrization of the projection on the $x-$space of the limit of $\zeroset{f}$ for $\varepsilon \rightarrow 0^+$.
\item
  {\sf UNION}. Given two parametrizations $Q_1,Q_2$, it returns a pa\-ra\-me\-tri\-za\-tion $Q$ such that $\zeroset{Q} = \zeroset{Q_1} \cup \zeroset{Q_2}$.
\end{itemize}

The input of \textsc{DegenerateSDP} is the $n-$variate $m \times m$ symmetric linear matrix $A(x)$ defining the spectrahedron $\spec(A)$, and a linear form $\ell$, that is supposed to be generic enough to satisfy the assumptions of Theorem \ref{thm:algo} below. The output is either a list $Q=[Q_1,\ldots,Q_{m-1}]$ of zero-dimensional parametrizations containing a solution $x^*$ to the original LMI (encoded as described above by $(Q,[a_*,b_*])$), or $(\,)$, in which case the original SDP \eqref{originalSDP} is either infeasible ($\spec(A)=\emptyset$) or the infimum in \eqref{originalSDP} equals $-\infty$.

Below we describe each step of the algorithm. 

\begin{algorithm}[H]
  \label{DegSDP}
  \begin{algorithmic}
    \Procedure{DegenerateSDP}{$A,\ell$}\label{step1}
    \State \label{step2} Generate $B \in \SS^m_+(\QQ)$
    \State \label{step4} $Q \leftarrow [\,]$
    \For{$r=1,\ldots,m-1$} \label{step5}
    \State \label{step6} $Q_r \leftarrow (1)$
    \For{$\iota \subset \{1,\ldots,m\}$ with $\sharp\iota=m-r$} \label{step7}
    \State \label{step8} $L \leftarrow \text{Lag}_{r,\iota}(A+\varepsilon B)$
    \State \label{step9} $Q_{r,\iota} \leftarrow {\sf ODP}(L,x)$
    \State \label{step10} $Q_r \leftarrow {\sf UNION}(Q_r,Q_{r,\iota})$
    \EndFor
    \State \label{step11} $Q \leftarrow [Q,{\sf CUT}(Q_r)]$
    \EndFor
    \State \new{{\bf if} {$\spec(A) \cap \zeroset{Q} = \emptyset$}
      {\bf then} {\Return{(\,)} }}
       \label{step12}
    \State \label{step13} \Return{$(Q,[a_*,b_*])$}
    \EndProcedure
  \end{algorithmic}
\end{algorithm}

Note that $\varepsilon$ in the previous formal description is treated as a variable, so that the polynomials in $L$ at step \ref{step8} define a curve. Remark that all solutions satisfy $\det A(x)=0$ hence $\rank A(x) \leq m-1$.


We show in Theorem \ref{thm:algo} that \textsc{DegenerateSDP} is correct and computes solutions to the original LMI as limits of perturbed solutions. We use the results of Sections \ref{sec:prelim} and \ref{sec:homotopy} and refer to the notation of Zariski open sets constructed in Corollary \ref{lemma:2} and \ref{lemma:3}, and in Proposition \ref{prop:smoothness} and \ref{prop:crit}.

\begin{theorem}\label{thm:algo}
  Let $A$ be a $m \times m$ $n-$variate symmetric linear matrix.
  Let $B \in \zarB_1 \cap \zarB_2 \cap \SS^m_+(\QQ)$, and $\ell \in \zarl_1 \cap \zarl_2 \cap \QQ^n$.
\begin{enumerate}
\item \label{A}
  If $A(x^*)=0$ for some $x^* \in \RR^n$, then $x^*$ is a minimizer in \eqref{originalSDP} or $\ell$ is unbounded from below on $\spec(A)$.
\item \label{B}
  Otherwise, $(Q,[a_*,b_*])=\mathrm{\textsc{DegenerateSDP}}(A,\ell)$ fulfils the following condition. If $x^* \in \spec(A)$ is a minimizer in \eqref{originalSDP} then $x^* \in \spec(A) \cap \zeroset{Q}$. Conversely, if $\spec(A) \neq \emptyset$, and $\ell$ is not unbounded from below on $\spec(A)$, then $\spec(A) \cap \zeroset{Q}$ contains a minimizer in \eqref{originalSDP}. \simone{ If the restriction of $\ell$ to $\spec(A)$ is proper, then $\zeroset{Q}$ contains the minimizers of the restriction of $\ell$ to $\spec(A+\eps B)$, for small enough $\eps>0$}.
\end{enumerate}
\end{theorem}
\begin{proof}
  First, suppose that $A(x^*)=0$ for some $x^* \in \RR^n$. Then $A_0=-\sum_ix^*_iA_i$, hence $A(x)=(x_1-x^*_1)A_1+\cdots+(x_n-x^*_n)A_n$. We deduce that $\spec(A)$ is the image under the translation $x \mapsto x+x^*$ of a cone, that is: either $\spec(A) = \{x^*\}$, in which case $\ell\equiv\ell(x^*)$ on $\spec(A)$, and $x^*$ is a minimizer for \eqref{originalSDP}, or $\spec(A)$ is an unbounded convex cone with origin in $x^*$. In the second case, since $\ell$ is linear, either its infimum on $\spec(A)$ is attained at the origin $x^*$, or its maximum is attained in $x^*$ and $\ell$ is unbounded from below on $\spec(A)$.

  We prove the first sentence in \eqref{B}. Assume that $x^* \in \spec(A)$ is a minimizer in \eqref{originalSDP}. Then $\ell(x^*)$ lies on the boundary of $\ell(\spec(A))$.
  By Lemma \ref{lemma:3}, we get that there exists $x^*_\varepsilon \in \spec(A + \varepsilon B)$ such that $\ell(x^*_\varepsilon)$ lies on the boundary of $\ell(\spec(A + \varepsilon B))$ and $\lim_\eps x^*_\eps = x^*$.   \simone{ Remark that by Lemma \ref{lemma:3} we also get that if the restriction of $\ell$ to $\spec(A)$ is proper, then the restriction of $\ext(\ell,\RR\la\eps\ra^n)$ to $\spec(A+\eps B)$ is proper, hence $\ell(\spec(A+\eps B))$ is closed. We deduce that $x^*_\varepsilon \in \ell(\spec(A + \varepsilon B))$.}
Hence for $\eps > 0$, $x^*_\varepsilon$ is a minimizer of $\ell$ on $\spec(A + \varepsilon B) \subset \RR^n$. By Proposition \ref{prop:crit}, there exists $r \in \{1,\ldots,m-1\}$, $\iota \subset \{1,\ldots,m\}$ with $\sharp\iota = m-r$, $y^*_\eps$ and $z^*_\eps$, such that $(x^*_\eps,y^*_\eps,z^*_\eps)$ is a solution of the Lagrange system $\text{Lag}_{r,\iota}(A+\varepsilon B)$. We deduce that for $\eps > 0$, $x^*_\eps$ is parametrized by the one-dimensinal parametrization $Q_{r,\iota} = {\sf ODP}(L)$ computed at step \ref{step9} of \textsc{DegenerateSDP}, hence by $Q_{r}$. We deduce that $Q$ parametrizes the limit $x^* = \lim_\eps x^*_\eps$, that is $x^* \in \spec(A) \cap \zeroset{Q}$.

  We finally come to the second sentence in \eqref{B}. Since $\ell$ is not unbounded on $\spec(A)$, and $\spec(A) \neq 0$, then the same holds for $\ell$ on $\spec(A+\varepsilon B)$. By Corollary \ref{lemma:2}, $\ell(\spec(A))$ and $\ell(\spec(A+\varepsilon B))$ are closed intervals. We deduce that the boundary of $\ell(\spec(A+\varepsilon B))$ is non-empty. Let $x^*_\varepsilon$ be such that $\ell(x^*_\varepsilon)$ lies on the boundary of $\ell(\spec(A+\varepsilon B))$. Since $\spec(A) \neq \emptyset$, by \ref{lemma:3} $x^* \coloneqq \lim_\varepsilon x^*_\varepsilon \in \spec(A) \cap \zeroset{Q}$ is such that $\ell(x^*)$ lies in the boundary of $\ell(\spec(A))$, hence a minimizer of SDP \eqref{originalSDP}.
\end{proof}

\new{To conclude, we make explicit the following fact that follows from Theorem \ref{thm:algo}. Recall that a generic linear form over a non-empty convex set is either unbounded from below ($\inf \ell = -\infty$) or its infimum is attained. Theorem \ref{thm:algo} implies that if $\ell$ is a generic linear form, then $\spec(A) \cap \zeroset{Q} = \emptyset$ if and only if $\spec(A) = \emptyset$ or $\ell$ is unbounded from below on $\spec(A)$. We conclude that up to genericity assumptions on the linear form, the algorithm is correct, since it returns a non-empty rational parametrization of a curve containing the minimizer if and only if problem \eqref{originalSDP} has a feasible solution.}




\subsection{Complexity analysis}\label{ssec:algo:compl}

This section contains a rigourous analysis of the arithmetic complexity of \textsc{DegenerateSDP}.
Let us first give an overview of the algorithms that are used to perform the subroutines in \textsc{DegenerateSDP}.

The computation of a one-dimensional parametrization of the homotopy curve $\zar{\mathcal{C}_{r,\iota}}$ at step
\ref{step9}, that is the routine {\sf ODP}, is done in two steps. First, we instantiate the system
$\text{Lag}_{r,\iota}(A+\varepsilon B)$ to a generic $\varepsilon=\overline{\varepsilon}$. By Proposition
\ref{prop:crit} we deduce that the obtained system is zero-dimensional. We use \cite{SaSc18} to compute a
zero-dimensional rational parametrization of this system.

The second steps consists in {\it lifting} the parameter $\varepsilon$ and in computing a {\it parametric
  geometric resolution} of $\text{Lag}_{r,\iota}(A+\varepsilon B)$ with the algorithm in \cite{Schost03}, that
is, a parametric analogue of \cite{GiLeSa01}. In our context, there is only one parameter, that is $\varepsilon$.

The routine {\sc CUT} can be performed via the algorithm in \cite{rrs} and, finally, the cost of the routine {\sc UNION} is given in \cite[Lem.G.3]{din2013nearly}.


To keep notations simple, let $L = (L_1,\ldots,L_N) \in \QQ[\varepsilon,t_1,\ldots,t_N]$ be the polynomials defining the Lagrange system \eqref{lagrange}, in the reduced form as in the proof of Proposition \eqref{prop:deg:curve}. Hence $N=c+n+r(m-r)$, where $c = {(m-r)(m+r+1)}/{2}$. The complex algebraic set $\zar{\mathcal{C}_{r,\iota}} = \zeroset{L}$ is a curve whose degree is bounded by Proposition \ref{prop:deg:curve}.

\begin{theorem}\label{theo:complexity}
  Let $L$ and $N$ be as above. Under the assumptions of Theorem \ref{thm:algo}, the output $Q=\textsc{DegenerateSDP}(A,\ell)$ is returned within
  $$
   \widetilde{O}\left( n \sum_r \binom{m}{r} r(m-r)N^4 \theta^2 \right)
  $$
    arithmetic operations over $\QQ$, where $\widetilde{O}(T)={O}(T\log^a(T))$ for some $a$ and $\theta \leq \binom{m^2+n}{n}^3$.
\end{theorem}
\begin{proof}
  Let $\overline{\varepsilon} \in (0,1)$ be generic, and let $\overline{L}$ be equal to the system $L$ where $\varepsilon$ is instantiated to
  $\overline{\varepsilon}$. Let $\theta$ be the value computed in \cite[Prop.5.1]{HNS2015c}, that bound the number of solution of $\overline{L}=0$
  in $\CC^N$. By the same proposition one gets
  $$
  \theta \leq \binom{c+n}{n}^3 \leq \binom{m(m-r)+n}{n}^3,
  $$
  from which the claimed bound uniform in $r$.

  Let $\overline{L}'=(\overline{L}'_1,\ldots,\overline{L}'_N)$ be a polynomial vector of lenght $N$ such that $\overline{L}'_i$ has
  the same multilinear structure as $\overline{L}_i$, for $i=1,\ldots,N$, and we denote by $H(T,t_1,\ldots,t_N) = T \overline{L} + (1-T)\overline{L}'$.
  By \cite[Prop.5]{SaSc18}, the complexity of computing a univariate representation of $\zeroset{\overline{L}}$ is in $\widetilde{O}(N^3 \theta\theta')$
  where $\theta' = \deg,\zeroset{H}$. By \cite[Lem.5.4]{HNS2015c}, $\theta' \in O(N\min\{n,c\}\theta)$. Hence the complexity of the first
  step of {\sf ODP} is in
  $$
  \widetilde{O}(\min\{n,c\} N^4 \theta^2).
  $$
  Next, let $\pi : \allowbreak\CC^{N+1} \rightarrow \CC$ be the projection $(\varepsilon,t_1,\ldots,t_N)\mapsto\varepsilon$. By \cite[Prop.5.1]{HNS2015c},
  a generic fiber of $\pi$ has degree bounded by $\theta$. Proposition \ref{prop:comparison} implies that $\deg\,\zar{\mathcal{C}_{r,\iota}}$ is bounded
  above by $\left(1+2r(m-r)\right) n \theta$. We apply the bound in \cite[Cor.1]{Schost03}, and we get a complexity in
  $$
  \widetilde{O}\left( nr(m-r) N^4 \theta^2 \right),
  $$
  for the parametric resolution step in {\sf ODP}. By \cite[Lem.13]{SaSc18}, the complexity of {\sc CUT} is in $\widetilde{O}\left( N^3 \theta\theta' \right)$,
  hence in
  $$
  \widetilde{O}\left( \min\{n,c\} N^4 \theta^2 \right).
  $$
  The complexity of {\sc UNION} is in $\widetilde{O}(N \theta^2)$ at each step, by \cite[Lem.G.3]{din2013nearly}. This shows that the most expensive step
  is the lifting step.

  The previous complexity bounds depend on $r$, and hold for all $r=1,\ldots,m$, and for all index subsets $\iota \subset \{1,\ldots,m\}$.
  We conclude by summing up with weight $\binom{m}{r}$, the number of subsets $\iota\subset\{1,\ldots,m\}$ of cardinality $m-r$.
\end{proof}

We note that $N$ can be bounded above by $n+2m^2$ uniformly in $r$. The complexity of \textsc{DegenerateSDP} given by Theorem \ref{theo:complexity} is
polynomial in $n$ when $m$ is fixed. Moreover, for a generic perturbation matrix $B$, \cite[Lem.3.1]{HNS2015c} allows to deduce the inequality
$n \geq \binom{m-r+1}{2}$: this implies that when $n$ is fixed, then $m$ is bounded above and hence the complexity is again a polynomial function of the input size.




\section{Examples}
\label{sec:examples}
In this final section we discuss degenerate examples, showing how our algorithm works in practice.

\begin{example}
Consider the $2 \times 2$ semidefinite representation of a point $(p_1,p_2) \in \RR^2$:
$$
\left\{ (x_1,x_2) \in \RR^2 \mymid A(x)\coloneqq \left(
\begin{array}{cc} p_1-x_1 & x_2-p_2 \\ x_2-p_2 & x_1-p_1 \end{array} \right) \succeq 0 \right\}
= \left\{ (p_1,p_2) \right\}.
$$
The interior of $\spec(A) \coloneqq \left\{ (p_1,p_2) \right\}$ in $\RR^2$ is empty, and moreover
$\spec(A)$, corresponding to the intersection of the $2-$dimensional linear space of matrices
in the pencil $A(x)$ with the $3-$dimensional cone of $2 \times 2$ symmetric matrices, has
co-dimension $2$ in $\RR^2$.

We first construct the incidence varieties $\calV_{r,\iota}(A)$. For $r=0$, the incidence variety
is smooth, but for $r=1$ and $\iota = \{1\}$, this is the following algebraic curve in $\CC^3$
$$
\calV_{1,\{1\}} = \zeroset{(x_2-p_2)y+p_1-x_1,(x_1-p_1)y+x_2-p_2}
$$
having two complex singularities lifting $(p_1,p_2)$, precisely at $(p_1,p_2,\allowbreak\pm {\mathfrak i})$, with ${\mathfrak i}^2=-1$.

According to Proposition \ref{prop:smoothness}, we can desingularize the varieties
$\calV_{r,\iota}(A)$ by applying a sufficiently generic homotopy
$$
A + \varepsilon B =
\left(\begin{array}{cc} p_1-x_1 & x_2-p_2 \\ x_2-p_2 & x_1-p_1 \end{array} \right)
+ \varepsilon \left(\begin{array}{cc} b_{11} & b_{12} \\ b_{12} & b_{22} \end{array} \right)
$$
perturbing the constant term of $A$. The set $\calV_{r,\iota}(A+ \varepsilon B)$ is smooth and
equidimensional for generic $B$, and the expected number of critical points of the restriction
of a generic linear function $\ell(x_1,x_2) = \ell_1x_1+\ell_2x_2$ is finite for each
$\varepsilon$.
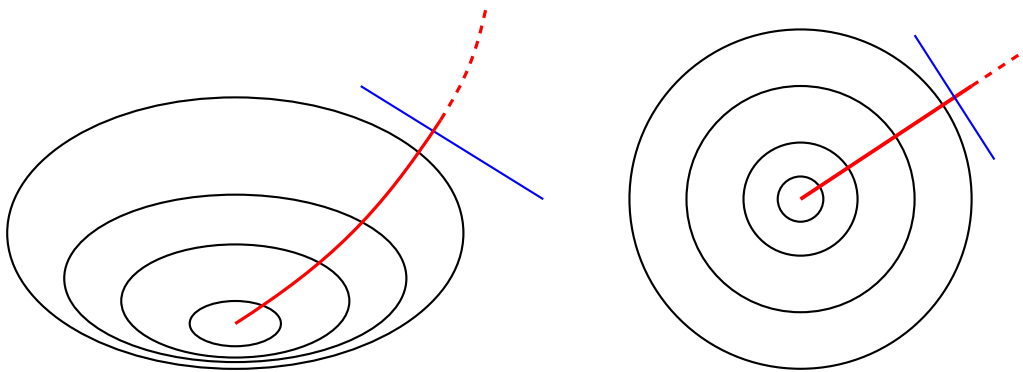
\begin{figure}[!ht]
\begin{center}
  \begin{tikzpicture}[scale=2.5]
  \tikzset{mypoints/.style={fill=white,draw=black,thick}}
  \def\ptsize{0pt}
  \draw[name path=ellipse,black,thick]
  (0,0) circle[x radius = 0.2 cm, y radius = 0.1 cm];
  \draw[name path=ellipse,black,thick]
  (0,0.1) circle[x radius = 0.5 cm, y radius = 0.25 cm];
  \draw[name path=ellipse,black,thick]
  (0,0.2) circle[x radius = 0.75 cm, y radius = 0.37 cm];
  \draw[name path=ellipse,black,thick]
  (0,0.4) circle[x radius = 1.0 cm, y radius = 0.6 cm];
  \draw[blue, thick] (1.35,0.55) -- (0.55,1.05);
  \draw[color=red, very thick] (0,0) to [bend right=12] (0.9,0.9);
  \draw[color=red, very thick, dashed] (0.9,0.9) to [bend right=10] (1.1,1.4);
\end{tikzpicture}
\qquad
\begin{tikzpicture}
\draw[thick] (0,0) circle (1.5);
\draw[thick] (0,0) circle (1.0);
\draw[thick] (0,0) circle (0.5);
\draw[thick] (0,0) circle (0.2);
\draw[line width=0.5mm, red] (0,0) -- (1.5,1);
\draw[red, dashed, very thick] (1.5,1) -- (2,4/3);
\draw[blue, thick] (1.7,0.35) -- (1,1.45);
\draw[decoration={text along path,reverse path,text align={align=center},text={}},decorate] (1.6,0) arc (0:180:1.6);
\draw[decoration={text along path,reverse path,text align={align=center},text={}},decorate] (1.05,0) arc (0:180:1.0);
\draw[decoration={text along path,reverse path,text align={align=center},text={}},decorate] (0.55,0) arc (0:180:0.6);
\end{tikzpicture}
\end{center}
\caption{Homotopy curves in red and linear objective function in blue, for
  generic $B$ (left) and for $B=\Id_2$ (right)}
\label{fig:1}
\end{figure}

In Figure \ref{fig:1} we plot the semialgebraic curve of solutions to the perturbed systems
for a fixed linear objective function.
Eliminating the variables $y$ and $z$ from the Lagrange
system $\text{Lag}_{r,\iota}(A+\varepsilon B)$ yields a one-dimensional complex curve
representing the Zariski closure of the red curves in Figure \ref{fig:1}.

For the special choice $B=\Id_2$, the real trace of the homotopy curve is the line
orthogonal to $\ell$, that is parallel to the zero set of $\ell^\perp(x_1,x_2) =
\ell_2x_1-\ell_1x_2$ and passing through $(p_1,p_2)$, whereas if $B$ is drawn randomly
the homotopy curve has degree $2$. For instance, for $(p_1,p_2)=(1,1)$, the homotopy
curve constructed by {\sc DegenerateSDP} is given by the equality
\begin{equation*}
\begin{aligned}
  & 2241769\,x_1^2 + 115046296\,x_1x_2 + 65669911\,x_2^2 - \\
  & - 119529834\,x_1 - 246386118\,x_2 + 182957976 = 0 \\
\end{aligned}
\end{equation*}
where $\ell(x_1,x_2) = 88x_1-94x_2$ is the objective function, and with
perturbation matrix
$$
B=\left(\begin{array}{cc} 80 & -68 \\ -68 & 109 \end{array} \right).
$$
We finally remark that, even if the choice $B = \Id_2$ exhibits a degenerate behaviour
in the sense described above, from the point of view of the homotopy constructed in
this work $B = \Id_2$ exhibits a generic behaviour: one can check by hand that the
incidence variety $\calV_{r,\iota}(A+ \varepsilon \Id_2)$ is singular if and only if
$\varepsilon=0$. Indeed, $\calV_{r,\iota}(A+ \varepsilon \Id_2)$ is defined by the
vanishing of $f^{(\varepsilon)}=(\varepsilon- x_1+x_2y,x_2+\varepsilon y+x_1y)$, and the
$2 \times 2$ minors of $Df^{(\varepsilon)}$ combined with $f^{(\varepsilon)}=0$ imply that
$y=\pm {\mathfrak i}$ and $0=x_2=\varepsilon-x_1=\varepsilon+x_1$ hence $x_1=x_2=
\varepsilon=0$.
\end{example}

\simone{
\begin{example}
  We consider the degenerate primal-dual semidefinite program in \cite[Ex.2.3.4]{drusvyatskiy2017many}.
  The original (dual) feasible set is the solution of the LMI
  $$
  A =
  \left(\begin{array}{ccc} 0 & 0 & 0 \\ 0 & 1 & 0 \\ 0 & 0 & 0 \end{array} \right) -
  x_1\left(\begin{array}{ccc} 0 & 0 & 0 \\ 0 & 0 & 0 \\ 0 & 0 & 1 \end{array} \right)-
  x_2\left(\begin{array}{ccc} 0 & 0 & 1 \\ 0 & 1 & 0 \\ 1 & 0 & 0 \end{array} \right)
  \succeq 0
  $$
  This is feasible but weakly (it defines a half line in $\RR^2$). As explained
  in \cite[Ex.2.3.4]{drusvyatskiy2017many}, this is a degenerate example since the minimum value
  of the linear objective function on the perturbed dual feasible sets $\spec(A + \varepsilon B)$
  might converge to the primal optimal value, rather than to the dual optimal value as expected.
  Indeed, the {\it duality gap} is positive for some choice of the objective function.

  This happens for instance for the linear objective function $\ell(x_1,x_2)=x_2$, that is in the SDP
  \begin{equation}
    \label{SDP_ex2}
    \begin{array}{cc}
     \text{ inf } & x_2 \\
     \text{ s.t. } & A(x_1,x_2) \succeq 0.
    \end{array}
  \end{equation}
  We apply our algorithm to problem \eqref{SDP_ex2}. We obtain that for a generic perturbation matrix $B$, the set of minimizers of the restriction of $x_2$ to $\spec(A+\eps B)$ lies on the complex curve of degree three given by the equation
  $$
  37467796356 \, x_1 \, x_2^2 - 101435939508 \, x_1 \, x_2 + 68653956361 \, x_1 = 0.
  $$
  Its real trace is the line $x_1=0$ : this is the real Zariski closure of the half
  line $\{(x_1,x_2) \in \RR^2 \mymid x_1=0, x_2 \geq 0\}$ containing all the minimizers
  $(0,\varepsilon)$ of the perturbed systems, and in particular the minimizer $(x_1,x_2) = (0,0)$
  of the SDP in \eqref{SDP_ex2}.
\end{example}
}


\section*{Acknowledgments}
Mohab Safey El Din is supported by the ANR grant ANR-17-CE40-0009 {\sc Galop} and the PGMO grant {\sc Gamma}. Simone Naldi acknowledges partial support of the Fondation Math\'ematique Jacques Hadamard throught the PGMO grant 2018-0061H.

\begin{center}
\rule{12cm}{0.7pt}
\end{center}

\bibliographystyle{ACM-Reference-Format}


\def\cfac#1{\ifmmode\setbox7\hbox{$\accent"5E#1$}\else
  \setbox7\hbox{\accent"5E#1}\penalty 10000\relax\fi\raise 1\ht7
  \hbox{\lower1.15ex\hbox to 1\wd7{\hss\accent"13\hss}}\penalty 10000
  \hskip-1\wd7\penalty 10000\box7}








\end{document}